\documentclass[webpdf,traditional,medium,namedate]{oup-authoring-template}
\onecolumn
\newif\iftikzX
\tikzXtrue
\tikzXfalse
\newif\ifFIGS 
\FIGSfalse    
\FIGStrue
\usetikzlibrary{external} 
\usepackage[T1]{fontenc}
\usepackage{amsmath,amssymb,mathtools,amsthm}
\usepackage{graphicx}
\usepackage{booktabs}
\usepackage{float}
\usepackage[section]{placeins}
\usepackage{array}
\usepackage{multirow}
\usepackage{url}
\usepackage{xcolor}
\usepackage{tikz}
\usepackage{pgfplots}
\usepackage{algpseudocode}
\usepackage[hidelinks]{hyperref}
\hypersetup{
  pdftitle={A Conditional-Distribution Framework for Validating Synthetic Multivariate Data},
  pdfauthor={Hari Dahal and Ishanu Chattopadhyay},
  pdfkeywords={conditional distributions, genomic surveillance, health data, model validation, real-world data, synthetic data}
}
\usetikzlibrary{shapes,calc,positioning,patterns,decorations.markings}
\usepgfplotslibrary{groupplots,statistics,colorbrewer}
\pgfplotsset{compat=1.18}

\newtheorem{theorem}{Theorem}[section]
\newtheorem{lemma}[theorem]{Lemma}
\newtheorem{corollary}[theorem]{Corollary}

\newcommand{\cgather}[2][\EQSP]{\begingroup\setlength\abovedisplayskip{#1}\setlength\belowdisplayskip{#1}\begin{gather} #2 \end{gather}\endgroup\noindent}

\newcommand{\calign}[2][\EQSP]{\begingroup\setlength\abovedisplayskip{#1}\setlength\belowdisplayskip{#1}\begin{align} #2 \end{align}\endgroup\noindent}

\begin{document}
\raggedbottom
\setlength{\emergencystretch}{2em}
\setlength{\textfloatsep}{12pt plus 2pt minus 2pt}
\setlength{\intextsep}{10pt plus 2pt minus 2pt}

\journaltitle{Biostatistics}
\copyrightyear{2026}
\pubyear{2026}
\firstpage{1}

\title[Conditional-distribution validation]{A Conditional-Distribution Framework for Validating Synthetic Multivariate Data}
\authormark{Dahal and Chattopadhyay}
\author[1]{Hari Dahal}
\author[1,$\ast$]{Ishanu Chattopadhyay}
\address[1]{University of Kentucky, Lexington, Kentucky, USA}
\corresp[$\ast$]{Corresponding author. Email: ishanu\_ch@uky.edu}

\abstract{Statistical validation of synthetic multivariate data requires assessing whether a generator preserves the joint dependence structure of the target population without merely reproducing observed records. We develop a model-agnostic framework based on full conditional distributions. For each coordinate, we normalize the conditional probability assigned to the observed value by the largest conditional probability available in the same record context; averaging this quantity yields a one-sided MAP-alignment statistic that can be estimated using a conditional model fitted on held-out real data. The mathematical contribution is twofold: under strict positivity and compatibility, the complete normalized conditional profile identifies the joint distribution, and its integrated $L^1$ difference defines a metric on finite-state generative processes; we also establish consistency and finite-sample concentration for the corresponding empirical estimators. Because high conditional alignment alone can arise from copying or concentration on conditional modes, we pair it with nearest-real similarity as a separate record-level novelty diagnostic. We evaluate the framework on NSHAP health and aging data, influenza B genomic surveillance, and 34 General Social Survey waves. In GSS, the Large Science Model matched the original-data control in mean conditional alignment while retaining substantial novelty, indicating preservation of conditional structure without row reuse. In influenza B, a Chow--Liu generator matched the control alignment but had almost no novelty, revealing near-reproduction of observed records. In NSHAP, alignment above the control together with unusually high novelty exposed concentration near conditional modes rather than improved fidelity. The framework therefore distinguishes three statistically different failure modes---loss of dependence, record reuse, and mode concentration---and provides a principled basis for validating synthetic health, surveillance, and population data.}

\maketitle

\section{Introduction}
Synthetic multivariate data are increasingly used in biomedical research, public-health surveillance, and population studies to support data sharing, method development, and analyses when direct access to individual-level records is restricted. Their validity is a joint-distribution problem. Conditional associations, interactions, subgroup distributions, multivariable risk patterns, and many target estimands depend on how variables vary together within records. A synthetic dataset can reproduce every marginal distribution, and may even perform well for a selected prediction task, while altering dependence structures that matter for other scientific analyses. Conversely, apparent fidelity can be achieved by reproducing source records rather than learning the population-generating process.

Most available evaluations do not directly target this distinction. Likelihood-based measures, including the evidence lower bound and perplexity, require tractable model probabilities and are not available for many generators \citep{kingma2014autoencoding,jelinek1977perplexity}. Task-based evaluations depend on a chosen outcome, estimand, and analysis model, so successful performance on one task does not establish general multivariate validity \citep{esteban2017realvalued}. Embedding distances and classifier two-sample tests depend on a representation that may not preserve scientifically relevant conditional structure \citep{heusel2017gans,lopezpaz2017revisiting}. Marginal moments and covariance matrices remain useful diagnostics, but they do not determine nonlinear or higher-order dependence \citep{snoke2018general,nowok2016synthpop}; Section~\ref{sec:examples} gives distributions with identical first- and second-order moments but different conditional behavior.

We formulate validation through full conditional distributions, which are statistically natural objects because, under strict positivity and compatibility, they identify the joint law \citep{brook1964,dobrushin1968,hammersley1971,besag1974}. The complete normalized conditional profile is therefore not an arbitrary utility score: we show that it retains identification of the underlying finite-state distribution and that its integrated $L^1$ difference defines a metric on generative processes. For sample-based evaluation, we derive a one-sided MAP-alignment statistic by fitting a conditional evaluator to real data and measuring, for each coordinate and record context, how closely the observed value agrees with the evaluator's conditional mode. We establish finite-sample concentration and consistency for empirical estimation under the stated conditions. This construction permits evaluation of generators without tractable likelihoods and without committing validation to a single downstream task.

MAP alignment measures structural agreement with a real-data evaluator, but it cannot by itself determine how that agreement was obtained. We therefore pair it with nearest-real similarity, which measures proximity to individual source records. The two quantities distinguish three failures that a single utility score can conflate: degradation of conditional dependence, near-reproduction of observed records, and overconcentration on conditional modes. The empirical studies are designed to illustrate these distinct regimes rather than to produce a universal ranking of generators. NSHAP provides a multidomain human-health setting, influenza B provides a genomic-surveillance setting in which near-identical records can occur naturally, and the GSS provides a heterogeneous high-dimensional stress test. The Large Science Model comparator uses a recursive conditional architecture previously applied to the infant microbiome and influenza evolution \citep{sizemore2024digital,wu2026emergenet}; the methodological contribution here is the statistical validation framework, not a claim of universal superiority for that generator class.

\section{Formal Setup for Conditional Analysis}
Let
\cgather{
X=(X^1,\dots,X^N)
}
take values in the finite product space
\cgather{
\mathcal{X}=\mathcal{X}^1\times\cdots\times\mathcal{X}^N.
}
Let $P$ be a strictly positive data-generating distribution. For coordinate
$i$, define the full conditional
\cgather{
P_i(x^i\mid x^{-i})
:=P(X^i=x^i\mid X^{-i}=x^{-i}).
}
A learned conditional kernel is denoted by
$\phi^i(\cdot\mid x^{-i})$. In practice, a small $\varepsilon$-floor may be
added and the conditional vector renormalized. Positivity is used together
with compatibility; it does not by itself guarantee that an arbitrary family
of kernels corresponds to a joint distribution.

\section{MAP-Alignment Functional}
For any model $\{\phi^i\}$ and sample $x\in\mathcal{X}$, define
\cgather{
    \upsilon(x,i)
    :=
    \frac{\phi^i(x^i\mid x^{-i})}
         {\max_{y\in \mathcal{X}^i}
          \phi^i(y\mid x^{-i})},
}
which equals $1$ exactly when $x^i$ is a maximizer of the model conditional.

Given a dataset $D=\{x_{k}\}_{k=1}^M$, define
\cgather{
    \Upsilon(D)
    :=
    \frac{1}{MN}
    \sum_{k=1}^M\sum_{i=1}^N
    \upsilon(x_{k},i),
}
an empirical estimate of
\cgather{
    \Upsilon_\phi(P)
    =
    \mathbb{E}_{X\sim P}
    \left[
        \frac{1}{N}\sum_{i=1}^N \upsilon(X,i)
    \right].
}
\subsection{Behavior Under Exact Conditionals}

\begin{lemma}[MAP-Alignment Under Exact Conditionals]
\label{lem:map}
Assume $\phi^i = P_i$ for all $i$.
Fix $i$ and $x^{-i}$, and let $p_j := P_i(j\mid x^{-i})$,
$p_{\max} := \max_j p_j$.
If $X^i \sim P_i(\cdot\mid x^{-i})$, then
\cgather{
    \upsilon(X,i) = \frac{P_i(X^i\mid x^{-i})}{p_{\max}},
\qquad
    \mathbb{E}[\upsilon(X,i)\mid x^{-i}]
    =
    \frac{1}{p_{\max}}\sum_j p_j^2.
}
Moreover:
\begin{itemize}
\item $\upsilon(X,i)=1$ iff $X^i\in\arg\max_j p_j$.
\item If $P_i(\cdot\mid x^{-i})$ is uniform on its support,
      then $\mathbb{E}[\upsilon(X,i)\mid x^{-i}] = 1$.
\item If some $p_j < p_{\max}$, then
      $\mathbb{E}[\upsilon(X,i)\mid x^{-i}] < 1$.
\end{itemize}
\end{lemma}

\begin{proof}
Immediate from the definition and the fact that
$\sum_j p_j^2 \le p_{\max}\sum_j p_j = p_{\max}$.
\end{proof}

Thus, under exact conditionals, $\upsilon(X,i)$ is a normalized
conditional-likelihood score. Its population target is
\cgather{
    \Upsilon_{\mathrm{oracle}}(P)
    :=
    \mathbb{E}_{X\sim P}
    \left[
        \frac{1}{N}\sum_{i=1}^N
        \frac{P_i(X^i\mid X^{-i})}
             {\max_y P_i(y\mid X^{-i})}
    \right].
    \label{eq:oracle_ratio}}
This target is generally smaller than one when the exact conditional is
nonuniform. Consequently, synthetic MAP alignment is assessed by agreement with a
real-data reference score, not by requiring $\Upsilon\to1$.

\section{Brook--Dobrushin Factorization and Identifiability}
Brook's lemma~\citep{brook1964} and related consistency results~\citep{dobrushin1968}
show how a compatible system of full conditionals identifies a strictly
positive joint distribution.

Assume
\cgather{
    \phi^i(x^i\mid x^{-i}) = P_i(x^i\mid x^{-i})
    \quad \text{whenever } P(x)>0.
}
Choose a reference configuration $x^\circ$ with $P(x^\circ)>0$.
Define the interpolating sequence
\cgather{
    x^{(0)} = x^\circ, \qquad
    x^{(i)} =
    (x^1,\dots,x^i, x^{i+1,\circ},\dots,x^{N,\circ}),
}
so $x^{(N)} = x$.
By repeated conditioning,
\cgather{
    \frac{P(x^{(i)})}{P(x^{(i-1)})}
    =
    \frac{P_i(x^i\mid x^{<i}, x^{>i,\circ})}
         {P_i(x^{i,\circ}\mid x^{<i}, x^{>i,\circ})}.
}
Multiplying yields the Brook factorization~\citep{brook1964}:
\cgather{
    \frac{P(x)}{P(x^\circ)}
    =
    \prod_{i=1}^N
    \frac{P_i(x^i\mid x^{<i}, x^{>i,\circ})}
         {P_i(x^{i,\circ}\mid x^{<i}, x^{>i,\circ})}.
}
Replacing $P_i$ by $\phi^i$ on the support gives an explicit
reconstruction of $P$ from the conditionals.

\begin{theorem}[Identification from compatible full conditionals]
\label{thm:brook}
Let $P$ and $Q$ be strictly positive distributions on the same finite product
space. If their full conditionals agree for every coordinate and every
configuration, then $P=Q$.
\end{theorem}

\begin{proof}
Brook's factorization reconstructs each joint probability ratio relative to a
fixed reference configuration from the full conditional system. Equality of
all full conditionals therefore gives identical probability ratios under $P$
and $Q$; normalization gives $P=Q$~\citep{brook1964,dobrushin1968}. The
compatibility assumption is essential: positivity alone does not guarantee that
an arbitrary collection of conditional kernels corresponds to a joint law.
\end{proof}

\section{Evaluating Generators via \texorpdfstring{$\Upsilon$}{Upsilon}}
For any model $\{\phi^i\}$,
\cgather{
    \Upsilon_\phi(P)
    =
    \mathbb{E}_{X\sim P}
    \left[
        \frac{1}{N}\sum_{i=1}^N
        \frac{\phi^i(X^i\mid X^{-i})}
             {\max_y \phi^i(y\mid X^{-i})}
    \right].
}
Let $D_{\mathrm{test}}$ be i.i.d.\ data from $P$. If $\phi^i_n\to P_i$
pointwise on the support and satisfy strict positivity, then dominated
convergence yields
\cgather{
    \Upsilon_{\phi_n}(P)
    \to
    \Upsilon_{\mathrm{oracle}}(P).
}
\begin{corollary}[Conditional-Kernel Convergence and MAP-Alignment Consistency]
Let $\phi^i_n$ be strictly positive kernels converging to $P_i$ on
$\mathrm{supp}(P)$. Let $\widetilde P_n$ denote the joint compatible
with $\{\phi^i_n\}$. Then:
\begin{itemize}
\item $\widetilde P_n \to P$ pointwise (equivalently, in total variation on the finite state space).
\item For any sequence of test sets $D_{\mathrm{test}}$ with
      $|D_{\mathrm{test}}|\to\infty$, the empirical scores
      $\Upsilon(D_{\mathrm{test}};\phi_n)$ converge in probability to
      $\Upsilon_{\mathrm{oracle}}(P)$.
\end{itemize}
Both conclusions follow from the assumed convergence of the full conditional
kernels. Convergence, or numerical agreement, of the scalar MAP-alignment score
alone does not imply convergence of the conditional kernels or recovery of the
joint law. Systematically low MAP alignment can nevertheless reveal disagreement
with the evaluator's conditional modes even when marginal summaries look
satisfactory.
\end{corollary}

\section{Conditional Inference, MAP-Alignment Uncertainty, and a Metric on Underlying Processes}

\subsection{Inference of Conditionals Using Conditional Learners}

Let $D=\{x^k\}_{k=1}^n$ be a training dataset in a finite product space. We
construct a family of full conditional models
$\Phi^{(n)}=\{\varphi^{i,n}(\cdot\mid x^{-i})\}_{i=1}^N$ by solving one
supervised prediction problem for each coordinate. In the experiments, these
models are implemented using conditional inference trees, whose unbiased
permutation-based split selection is described in
\citep{hothorn2006,hothorn2006party,strobl2007}.

The theory does not require a tree-specific convergence exponent. We instead
state the estimation assumption explicitly: for each coordinate,
\begin{equation}
\begin{aligned}
&\sup_{x^{-i},x^i}
\left|\varphi^{i,n}(x^i\mid x^{-i})-P_i(x^i\mid x^{-i})\right|
\leq \delta_{i,n},\\
&\delta_n:=\max_i\delta_{i,n}\xrightarrow{p}0.
\end{aligned}
\label{eq:conditional_error}
\end{equation}
A small $\varepsilon$-floor may be applied before renormalization to maintain
strict positivity on the effective support. When the learned kernels are
compatible, the resulting full conditional system identifies a unique joint
law by Theorem~\ref{thm:brook}.

\subsection{MAP-Alignment and Its Finite-Sample Uncertainty}

Given a trained conditional family $\Phi^{(n)}$, define
\cgather{
\upsilon_{\Phi^{(n)}}(x,i)
=
\frac{\varphi^{i,n}(x^i\mid x^{-i})}
{\max_{y\in\mathcal{X}^i}\varphi^{i,n}(y\mid x^{-i})}
\in[0,1].
\label{eq:upsilon_definition}
}
For an independent test dataset $D_{\mathrm{test}}=\{x_k\}_{k=1}^M$, let
\cgather{
Z_k:=\frac{1}{N}\sum_{i=1}^N
\upsilon_{\Phi^{(n)}}(x_k,i),
\qquad
\hat\Upsilon_{\Phi^{(n)}}(D_{\mathrm{test}})
:=\frac{1}{M}\sum_{k=1}^M Z_k.
\label{eq:upsilon_empirical}
}
The coordinates within a record may be dependent, but the row scores $Z_k$ are
independent and bounded in $[0,1]$ under i.i.d.\ test sampling. Hence
Hoeffding's inequality~\citep{hoeffding1963} gives
\cgather{
\Pr\left(
\left|\hat\Upsilon_{\Phi^{(n)}}-\mathbb{E}Z_1\right|>\epsilon
\right)
\leq 2\exp(-2M\epsilon^2).
\label{eq:upsilon_hoeffding}
}
Because each finite conditional vector has modal probability at least
$1/|\mathcal{X}^i|$, the normalization in
Eq.~\eqref{eq:upsilon_definition} is Lipschitz in the conditional vector. Combining test-sample variation with
Eq.~\eqref{eq:conditional_error} yields
\cgather{
\hat\Upsilon_{\Phi^{(n)}}(D_{\mathrm{test}})
=
\Upsilon_{\mathrm{oracle}}(P)
+O_p\!\left(M^{-1/2}+\delta_n\right).
\label{eq:upsilon_full_uncertainty}
}

\subsection{A Metric on Underlying Generative Processes}

For a strictly positive process $P$ with full conditionals $P_i$, define its population $\upsilon$-profile:
\cgather{
u_P(x,i)
=
\frac{
P_i(x^i\mid x^{-i})
}{
\max_{y\in\mathcal{X}^i} P_i(y\mid x^{-i})
}.
}
Fix a reference probability measure $\mu$ on $\mathcal{X}$ with full
support, and sample the coordinate index uniformly from
$\{1,\dots,N\}$. Define the distance
\cgather{
d_{\mu}(P,Q)
=
\mathbb{E}_{X\sim\mu}
\left[
\frac{1}{N}\sum_{i=1}^{N}
\lvert u_{P}(X,i) - u_{Q}(X,i)\rvert
\right].
\label{eq:true_metric}
}

\begin{theorem}
Let $\mu$ have full support on the finite product space $\mathcal{X}$.
Then $d_{\mu}$ is a metric on the set of strictly positive probability
distributions on $\mathcal{X}$.
\end{theorem}

\begin{proof}
Nonnegativity and symmetry are immediate. The triangle inequality follows from the scalar inequality
$\lvert a-c\rvert \le \lvert a-b\rvert + \lvert b-c\rvert$ applied pointwise and averaged over coordinates and integrated with respect to $\mu$.
For identity of indiscernibles, if $d_{\mu}(P_1,P_2)=0$, then
$u_{P_1}(x,i)=u_{P_2}(x,i)$ for every $(x,i)$ because $\mathcal{X}$ is
finite and $\mu$ has full support. For each fixed
$(x^{-i},i)$, the conditional distribution is recovered from its normalized
profile by
$P_i(x^i\mid x^{-i})=u_P(x,i)/\sum_{z\in\mathcal{X}^i}u_P((z,x^{-i}),i)$.
Thus all full conditionals agree. Theorem~\ref{thm:brook} then implies $P_1=P_2$.
\end{proof}

Given empirical datasets $D_1,D_2$, train conditional generator families
$G_1,G_2$ yielding conditional families $\Phi^{(1)}$ and $\Phi^{(2)}$.
Independently draw a common evaluation sample
$E_{\mu}=\{x^{(k)}\}_{k=1}^{M}$ i.i.d.\ from the fixed reference measure
$\mu$. Define the empirical distance
\cgather{
\widehat d_{\mu}(D_1,D_2;E_{\mu})
=
\frac{1}{MN}
\sum_{k=1}^{M}
\sum_{i=1}^{N}
\left\lvert
u_{G_1}(x^{(k)},i)-u_{G_2}(x^{(k)},i)
\right\rvert.
\label{eq:empirical_metric}
}

For notational convenience, let $d_{\mu}(G_1,G_2)$ denote
Eq.~\eqref{eq:true_metric} evaluated using the two learned normalized profiles.

\begin{theorem}
Assume the conditional estimators satisfy \eqref{eq:conditional_error} with
uniform errors $\epsilon_1,\epsilon_2$. Then
\cgather{
\big|
\widehat d_{\mu}(D_1,D_2;E_{\mu})-d_{\mu}(P_1,P_2)
\big|
=
O_p\!\left(
M^{-1/2}
+
\epsilon_1
+
\epsilon_2
\right).
}
\end{theorem}

\begin{proof}
By the triangle inequality,
\calign{
\left\lvert \widehat d_{\mu}-d_{\mu}(P_1,P_2)\right\rvert
&\leq
\left\lvert \widehat d_{\mu}-d_{\mu}(G_1,G_2)\right\rvert\\
&\quad+
\left\lvert d_{\mu}(G_1,G_2)-d_{\mu}(P_1,P_2)\right\rvert.
}
For the first term, the row-level quantities
\cgather{
Z_k=\frac{1}{N}\sum_{i=1}^{N}
\left\lvert u_{G_1}(x^{(k)},i)-u_{G_2}(x^{(k)},i)\right\rvert
}
are i.i.d.\ and bounded in $[0,1]$ because $E_{\mu}$ is independent of the
training datasets. Hoeffding's inequality therefore gives the
$O_p(M^{-1/2})$ term. The second term is controlled by the Lipschitz
continuity of $u(\varphi)=\varphi/\max_y\varphi(y)$ on finite
probability simplices, together with \eqref{eq:conditional_error}.
\end{proof}
\subsection{One-Sided MAP-Alignment Comparison Against Held-Out Real Data}

Let $D_{\mathrm{fit}}$ be an evaluator-training dataset independent of both
the real reference dataset $D_{\mathrm{ref}}$ and the synthetic dataset
$D_{\mathrm{syn}}$. Fit a conditional family $\Phi^{(\mathrm{real})}$ on
$D_{\mathrm{fit}}$, and evaluate both the real reference records and synthetic
records under that same fixed conditional family:
\cgather{
\Upsilon_{\mathrm{ref}}
=\hat\Upsilon_{\Phi^{(\mathrm{real})}}(D_{\mathrm{ref}}),
\qquad
\Upsilon_{\mathrm{syn}}
=\hat\Upsilon_{\Phi^{(\mathrm{real})}}(D_{\mathrm{syn}}).
}
We report the signed MAP-alignment gap and retention ratio
\cgather{
\Delta\Upsilon:=\Upsilon_{\mathrm{syn}}-\Upsilon_{\mathrm{ref}},
\qquad
R_{\Upsilon}:=\frac{\Upsilon_{\mathrm{syn}}}{\Upsilon_{\mathrm{ref}}}.
\label{eq:fidelity_gap_retention}
}
A generator matching the reference in average MAP alignment should have
$\Delta\Upsilon$ near zero, or equivalently $R_{\Upsilon}$ near one. Negative
gaps indicate lower average MAP alignment;
positive gaps do not by themselves establish superiority because finite-sample
estimation and mode concentration can increase the normalized score. In
particular, a generator that overrepresents conditional modal values can score
above the original-data control; $\Upsilon$ is therefore a comparative
diagnostic rather than a stand-alone objective to maximize. When an
independent evaluator-training dataset is unavailable, cross-fitting may instead
be used to avoid evaluating real records under a model trained on those same
records.

\subsection{Consistency of the Dataset Distance as a Metric on Processes}
\begin{theorem}[Convergence to the Population Metric]
Let $D_1$ and $D_2$ be drawn i.i.d.\ from strictly positive processes
$P_1$ and $P_2$, and let $E_{\mu}$ be an independent i.i.d.\ sample from the
fixed full-support reference measure $\mu$. Suppose the learned conditional kernels $\widehat P_{k,i}$ satisfy
\calign{
&\max_{1\leq i\leq N}\sup_{x^{-i},x^i}
\left\lvert
\widehat P_{k,i}(x^i\mid x^{-i})-P_{k,i}(x^i\mid x^{-i})
\right\rvert\\
&\hspace{5em}\leq \epsilon_k,
\qquad k\in\{1,2\}.
}
with $\epsilon_k\to 0$ in probability as $|D_k|\to\infty$, and suppose
$M=|E_{\mu}|\to\infty$. Then
\cgather{
\widehat d_{\mu}(D_1,D_2;E_{\mu})
\xrightarrow{p}
d_{\mu}(P_1,P_2).
\label{eq:limit_metric}
}
\end{theorem}

\begin{proof}
From the previous theorem we have
\cgather{
\big|
\widehat d_{\mu}(D_1,D_2;E_{\mu})-d_{\mu}(P_1,P_2)
\big|
=
O_p\!\left(
M^{-1/2}
+
\epsilon_1
+
\epsilon_2
\right).
}
As $M\to\infty$ and $|D_k|\to\infty$, both $M^{-1/2}$ and $\epsilon_k$ converge to zero in probability, which implies \eqref{eq:limit_metric}.
\end{proof}

\begin{corollary}
If $P_1=P_2$, then
\cgather{
\widehat d_{\mu}(D_1,D_2;E_{\mu})\xrightarrow{p}0.
}
If $P_1\neq P_2$, then
\cgather{
\widehat d_{\mu}(D_1,D_2;E_{\mu})
\xrightarrow{p}d_{\mu}(P_1,P_2)>0.
}
\end{corollary}

\subsection{Cross-Fitted Evaluation When Independent Training Data Are Unavailable}

Algorithm~\ref{alg:crossfit} gives a general cross-fitted implementation for
settings in which a separate evaluator-training dataset is unavailable. All
three empirical applications reported here used fixed conditional evaluators
trained on data disjoint from the corresponding Original (control) and
synthetic records. Algorithm~\ref{alg:crossfit} was therefore not invoked in the
reported experiments and is provided for applications without an independent
evaluator-training partition.

\begin{algorithm}[t]
\caption{Cross-fitted conditional MAP alignment and row novelty}
\label{alg:crossfit}
\begin{algorithmic}[1]
\Require Real data $D_{\mathrm{real}}$; synthetic data $D_{\mathrm{syn}}$; $K$ folds; conditional learner $\mathsf{Learn}$; level $\alpha$
\Ensure $\hat\Upsilon_{\mathrm{ref}}$, $\hat\Upsilon_{\mathrm{syn}}$, $\Delta\hat\Upsilon$, $R_{\Upsilon}$, and row novelty
\State Partition $D_{\mathrm{real}}$ into folds $F_1,\dots,F_K$
\For{$k=1,\dots,K$}
  \State Fit $\Phi_k\leftarrow\mathsf{Learn}(D_{\mathrm{real}}\setminus F_k)$
  \State Score each $x\in F_k$ using $\Phi_k$ and store its row-average MAP alignment
  \State Score an assigned subset of $D_{\mathrm{syn}}$ using the same $\Phi_k$
\EndFor
\State Average the stored real and synthetic row scores to obtain $\hat\Upsilon_{\mathrm{ref}}$ and $\hat\Upsilon_{\mathrm{syn}}$
\State Set $\Delta\hat\Upsilon\leftarrow\hat\Upsilon_{\mathrm{syn}}-\hat\Upsilon_{\mathrm{ref}}$ and $R_{\Upsilon}\leftarrow\hat\Upsilon_{\mathrm{syn}}/\hat\Upsilon_{\mathrm{ref}}$
\State For each synthetic row, compute nearest-real similarity $\eta_j$ and row novelty $1-m^{-1}\sum_j\eta_j$
\State Use a paired bootstrap over independent datasets or survey waves for the final interval
\State \Return all summary statistics
\end{algorithmic}
\end{algorithm}

For a fixed fitted conditional family and $M$ independent evaluated rows, the
row-level Hoeffding radius at confidence $1-\delta$ is
\begin{equation}
r_M(\delta)=\sqrt{\frac{1}{2M}\log\frac{2}{\delta}}.
\end{equation}
At $\delta=0.05$, the radii are $0.1358$ for $M=100$ and $0.0429$ for
$M=1000$. Conditional-estimation uncertainty must be handled separately through
data splitting, repeated fitting, or a higher-level bootstrap.

\section{Low-Order Moment Matching is Insufficient: Two Illustrative Examples}\label{sec:examples}

Two examples in $\mathbb{R}^3$ show that identical first- and second-order
moments do not imply similar conditional structure. We use the
continuous-density analogue of the finite-state score, replacing the maximum by
an essential supremum. A formal extension to continuous and mixed-type data is
left for future work.

\subsection{Example 1: Uniform vs.\ Gaussian with Identical Moments}

Consider the following two distributions on $\mathbb{R}^3$:
\cgather{
X^{(U)} \sim \mathrm{Uniform}([-\sqrt{3},\sqrt{3}]^3),
\qquad
X^{(G)} \sim \mathcal{N}_3(0,I_3).
}

Both satisfy
\cgather{
\mathbb{E}[X^{(U)}]
=
\mathbb{E}[X^{(G)}]
=
0,
\qquad
\mathrm{Cov}(X^{(U)})
=
\mathrm{Cov}(X^{(G)})
=
I_3.
}
Hence all marginal means, variances, and the full covariance matrix
coincide.

Yet their conditional structures differ sharply. For $X^{(U)}$, each
coordinate has a flat conditional density on $[-\sqrt{3},\sqrt{3}]$,
yielding
\cgather{
\upsilon(X^{(U)}, i) = 1
\quad\text{for almost every sample}.
}
For $X^{(G)}$, the $i$th coordinate conditional is
$X^{(G)}_i \mid X^{(G)}_{-i} \sim \mathcal{N}(0,1)$, giving
\cgather{
\upsilon(X^{(G)}, i)
=
\exp\!\left(-\tfrac12 X^{(G)}_i{}^2\right),
}
with mean exactly $1/\sqrt{2}\approx0.7071$. Thus, although the datasets match in mean
and covariance, the MAP-alignment values differ substantially:
\cgather{
\Upsilon(X^{(U)}) = 1,
\qquad
\Upsilon(X^{(G)}) = \frac{1}{\sqrt{2}}.
}
Low-order moments fail to detect this discrepancy.

\subsection{Example 2: Matching Non-Identity Covariances with Distinct Conditional Structure}

To demonstrate that the limitation persists even when the covariance
matrix is nontrivial, apply the same invertible linear map
\cgather{
L
=
\begin{pmatrix}
1 & \rho & 0\\
0 & 1     & \rho\\
0 & 0     & 1
\end{pmatrix},
\qquad 0 < \rho < 1,
}
to both datasets and define
\cgather{
Y^{(U)} = L X^{(U)},
\qquad
Y^{(G)} = L X^{(G)}.
}

Both transformed datasets have the \emph{same} non-identity covariance
matrix
\cgather{
\mathrm{Cov}(Y^{(U)})
=
\mathrm{Cov}(Y^{(G)})
=
\Sigma
=
L L^{\!\top},
}
and share identical columnwise means and variances.

Their conditional distributions nevertheless differ. Since $Y^{(U)}$ is the
image of a uniform cube under a linear shear, each
conditional $Y^{i,(U)} \mid Y^{-i,(U)}$ is uniform on a finite
interval (given by the intersection of a line with a parallelepiped),
implying
\cgather{
\upsilon(Y^{(U)}, i) = 1
\quad\text{for almost every sample}.
}

Conversely, $Y^{(G)} \sim \mathcal{N}_3(0,\Sigma)$ has
linear--Gaussian conditionals. If
$Y^{i,(G)} \mid Y^{-i,(G)} \sim \mathcal{N}(m_i(y^{-i}), \sigma_i^2)$,
then
\cgather{
\upsilon(Y^{(G)}, i)
=
\exp\!\left(
-\tfrac12 (y^i - m_i(y^{-i}))^2 / \sigma_i^2
\right),
}
and the standardized conditional residual is $\mathcal{N}(0,1)$, so its mean is exactly $1/\sqrt{2}\approx0.7071$.

Thus, even though $Y^{(U)}$ and $Y^{(G)}$ agree in all first- and
second-order statistics, their conditional behavior differs markedly,
and the MAP-alignment statistic again reveals the discrepancy:
\cgather{
\Upsilon(Y^{(U)}) = 1,
\qquad
\Upsilon(Y^{(G)}) = \frac{1}{\sqrt{2}}.
}

\subsection{Implications}

Thus, matching means and covariance does not determine conditional structure.
MAP alignment detects the discrepancy in both examples.

\section{Statistical diagnostics for synthetic samples}\label{sec:applications}

\subsection{Row-Level Memorization Diagnostic}

MAP alignment does not detect row reuse: a dataset can score highly while
copying observed records. We therefore add a row-level novelty diagnostic.

Let $D_{\mathrm{real}}=\{x^{(1)},\ldots,x^{(n)}\}$ and
$D_{\mathrm{syn}}=\{y^{(1)},\ldots,y^{(m)}\}$ be real and synthetic datasets
defined over the same categorical (or discretized) feature space. Using a
concatenated one-hot embedding $f(\cdot)$, we define cosine similarity
\cgather{
s(u,v)
\;=\;
\frac{\langle f(u), f(v)\rangle}
{\|f(u)\|_2\,\|f(v)\|_2},
} which, for categorical data, reduces to the fraction of coordinates on which two
records agree.

For each synthetic record $y^{(j)}$, we define its nearest-real similarity
\cgather{
\eta_j
\;=\;
\max_{1 \le i \le n}
s\!\left(x^{(i)}, y^{(j)}\right).
}
We summarize row-level reuse via the \emph{lack-of-novelty score}
\cgather{
\mathcal{N}(D_{\mathrm{syn}})
\;=\;
\frac{1}{m}\sum_{j=1}^m \eta_j,
}
with larger values indicating lower novelty. In particular,
$\mathcal{N}=1$ corresponds to exact row reuse (e.g., bootstrap resampling),
while smaller values indicate increasing deviation from any individual real
record. Tail statistics of $\{\eta_j\}$ (e.g., the 95th percentile)
capture worst-case near-copying behavior.

The diagnostics are complementary: $\Upsilon(D)$ measures conditional
alignment, whereas $\mathcal{N}(D_{\mathrm{syn}})$ measures proximity to
specific real records. Nearest-row novelty is a memorization diagnostic, not a
formal privacy guarantee; disclosure risk requires separate membership, linkage,
or attribute-inference analysis. The target is $\Delta\Upsilon\to0$
(equivalently $R_{\Upsilon}\to1$) with $\mathcal{N}$ well below one.

\subsection{Practical Use of the Diagnostics}

The diagnostics answer different questions and should be selected accordingly.
MAP alignment is appropriate when the primary concern is whether synthetic
records preserve the conditional response structure represented by a fixed
real-data evaluator. To keep reporting comparable, all empirical applications
below use the same summaries: mean $\Upsilon$, the signed gap from the original
control, the retention ratio, nearest-real similarity, and row novelty. The GSS
analysis additionally displays the zero-anchored excess-alignment fraction
$\Lambda$ as a secondary normalization because its independent-marginal null is
well separated from the control.

Nearest-real similarity is appropriate when the concern is copying or near reuse
of individual records. It should be reported together with exact duplicate rates
and upper-tail summaries, such as the 95th or 99th percentile, rather than only
a mean. It is not a substitute for formal disclosure-risk analysis.

Neither diagnostic replaces marginal checks or application-specific validation.
When the synthetic data will support a particular estimand, classifier, or policy
simulation, MAP alignment and novelty should be combined with relevant marginal,
task-based, and, where necessary, privacy metrics. A single scalar score is not
sufficient for all uses of synthetic multivariate data.

\section{Empirical studies}
\label{sec:empirical_results}

The empirical analyses examine whether the proposed diagnostics separate the three inferentially distinct departures identified above. The data sources were selected to provide different challenges to synthetic-data validation rather than exchangeable replications of a common experiment. NSHAP represents heterogeneous human-health and social measurements, for which valid synthesis must retain clinically and behaviorally structured dependence while accommodating substantial nonmodal variation. Influenza B sequences represent genomic surveillance, where strong local dependence and naturally recurring sequence patterns make conditional alignment difficult to distinguish from record reuse. The GSS provides a repeated-cross-sectional stress test with substantial variation in dimension and survey content across waves. Analyses and uncertainty summaries were therefore constructed within each source, and no pooled inferential comparison was formed across applications. The data sources and corresponding analysis units are summarized in Table~\ref{tab:data_blocks}.

\begin{table}[H]
\centering
\caption{Data sources and analysis units used in the empirical evaluation.}
\label{tab:data_blocks}
\small
\setlength{\tabcolsep}{3.5pt}
\resizebox{\textwidth}{!}{%
\begin{tabular}{lcccc}
\toprule
Data source & Analysis units & Coverage & Variables per unit & Records evaluated per condition \\
\midrule
General Social Survey & 34 waves & 1972--2024 & 197--1,273 & 1,000 \\
Influenza B surveillance & 7 overlapping windows & 2016--2023 & 220--374 & 1,000 \\
NSHAP Round 2 & 8 archived analyses & 2010--2011 & 861--862 & 1,000 \\
\bottomrule
\end{tabular}%
}
\end{table}

\subsection{Common evaluation protocol}

Within each analysis unit, a fixed conditional evaluator was fitted using records that were not used as the original-data control or as synthetic evaluation records. We compared the Large Science Model (LSM), a restricted Chow--Liu hybrid \citep{chow1968}, CTGAN \citep{xu2019ctgan}, an independent-marginal baseline, and an original-data control. The Chow--Liu comparator fitted a maximum-weight dependence tree to a bounded subset of lower-cardinality coordinates and sampled the remaining coordinates from their empirical marginal distributions. CTGAN used SDV single-table metadata inference for categorical or discretely represented variables.

Each synthetic condition generated 1,000 records. The original-data control was a simple random sample of 1,000 source records, drawn without replacement when the source contained at least 1,000 records and with replacement otherwise. All conditions within an analysis unit were evaluated at the same sample size. Missing values were not statistically imputed or recoded as substantive response levels. A missing target coordinate was excluded from that record's MAP-alignment average, while missing predictor coordinates remained unspecified to the conditional learner. For nearest-record similarity, missing entries were represented by a common marker only for one-hot encoding.

For every application, we report mean MAP alignment, the signed difference from the original-data control, the retention ratio, mean nearest-record similarity, and row novelty. The last quantity is one minus mean nearest-record similarity. Intervals for GSS are nonparametric wave-bootstrap intervals. The influenza B windows overlap, and the NSHAP runs share a common data source; intervals for those analyses are therefore descriptive run-resampling summaries rather than population-level confidence intervals.

\subsection{General Social Survey}
\label{sec:gss_results}

The GSS is a repeated cross-sectional survey of adults in the United States \citep{gss2024}. We used it as a heterogeneous high-dimensional stress test because the number and composition of measured variables change substantially across waves, while the records retain complex conditional relationships among demographic characteristics, experiences, and attitudes. We analyzed 34 waves collected between 1972 and 2024, comprising 71,667 respondent records and 197--1,273 modeled variables per wave. Within each wave, 50 percent of respondents were used to fit the wave-specific LSM and conditional evaluator; the remaining 50 percent were reserved for comparator construction, the original-data control, and evaluation. The characteristics of the GSS analysis are summarized in Table~\ref{tab:gss_data}.

\begin{table}[H]
\centering
\caption{Characteristics of the General Social Survey analysis.}
\label{tab:gss_data}
\small
\setlength{\tabcolsep}{3.5pt}
\begin{tabular}{@{}lr@{}}
\toprule
Characteristic & Value \\
\midrule
Analyzed waves & 34 \\
Survey years & 1972--2024 \\
Respondent records & 71,667 \\
Respondents per wave & 1,372--4,510 \\
Evaluator and LSM fitting partition & 50 percent per wave \\
Held-out analysis partition & 50 percent per wave \\
Modeled variables per wave & 197--1,273 \\
Evaluated records per condition & 1,000 \\
\bottomrule
\end{tabular}
\end{table}

The aggregate GSS benchmark results across the 34 survey waves are summarized in Table~\ref{tab:gss_benchmark}.

\begin{table}[H]
\centering
\caption{General Social Survey results across 34 waves. The signed gap is the generator mean MAP alignment minus the original-data control mean. Retention is their ratio, and row novelty is one minus mean nearest-record similarity. Intervals are nonparametric 95 percent wave-bootstrap intervals.}
\label{tab:gss_benchmark}
\small
\setlength{\tabcolsep}{3.2pt}
\resizebox{\textwidth}{!}{%
\begin{tabular}{lccccc}
\toprule
Generator & Mean MAP alignment & Signed gap (95 percent interval) & Retention & Nearest-record similarity & Row novelty \\
\midrule
LSM & 0.8032 & 0.0027 $[-0.0004,0.0062]$ & 1.0034 & 0.3800 & 0.6200 \\
Chow--Liu hybrid & 0.7872 & $-0.0132$ $[-0.0159,-0.0106]$ & 0.9834 & 0.6560 & 0.3440 \\
CTGAN & 0.7630 & $-0.0375$ $[-0.0415,-0.0332]$ & 0.9530 & 0.5665 & 0.4335 \\
Independent baseline & 0.7567 & $-0.0437$ $[-0.0483,-0.0390]$ & 0.9451 & 0.5619 & 0.4381 \\
Original-data control & 0.8004 & 0.0000 & 1.0000 & 1.0000 & 0.0000 \\
\bottomrule
\end{tabular}%
}
\end{table}

As shown in Table~\ref{tab:gss_benchmark}, the mean MAP alignment for LSM was 0.8032, compared with 0.8004 for the original-data control. The mean signed difference was 0.0027, with a 95 percent wave-bootstrap interval from -0.0004 to 0.0062, and the retention ratio was 1.0034. The interval includes zero, indicating that the mean LSM score was compatible with the control at the resolution of the wave-level analysis. This comparison does not establish equality of the complete conditional systems. The independent-marginal baseline retained 0.9451 of the control score, indicating a systematic reduction in alignment after cross-coordinate dependence was removed.

For the GSS analysis only, we also display a zero-anchored normalization relative to the independent-marginal baseline,
\begin{equation}
\Lambda_g =
\frac{\bar{\Upsilon}_g-\bar{\Upsilon}_{\mathrm{ind}}}
{\bar{\Upsilon}_{\mathrm{control}}-\bar{\Upsilon}_{\mathrm{ind}}}.
\label{eq:excess_alignment}
\end{equation}
The LSM estimate was 1.062, with a wave-bootstrap interval from 0.992 to 1.154; the interval includes the control anchor of one. LSM row novelty was 0.6200, compared with 0.3440 for Chow--Liu, 0.4335 for CTGAN, and 0.4381 for the independent baseline. Figure~\ref{fig:gss_benchmark_four_panel} shows the corresponding wave-level behavior: LSM signed gaps remain concentrated near zero and its retention remains close to one across survey years, whereas the other synthetic generators are generally shifted below the original-data control. The figure also shows that LSM maintains the highest row novelty across waves.

\begin{figure}[t]
  \centering
  \tikzexternalenable
  \tikzsetnextfilename{figres}
  \iftikzX
\else
  \includegraphics[width=1.65\textwidth]{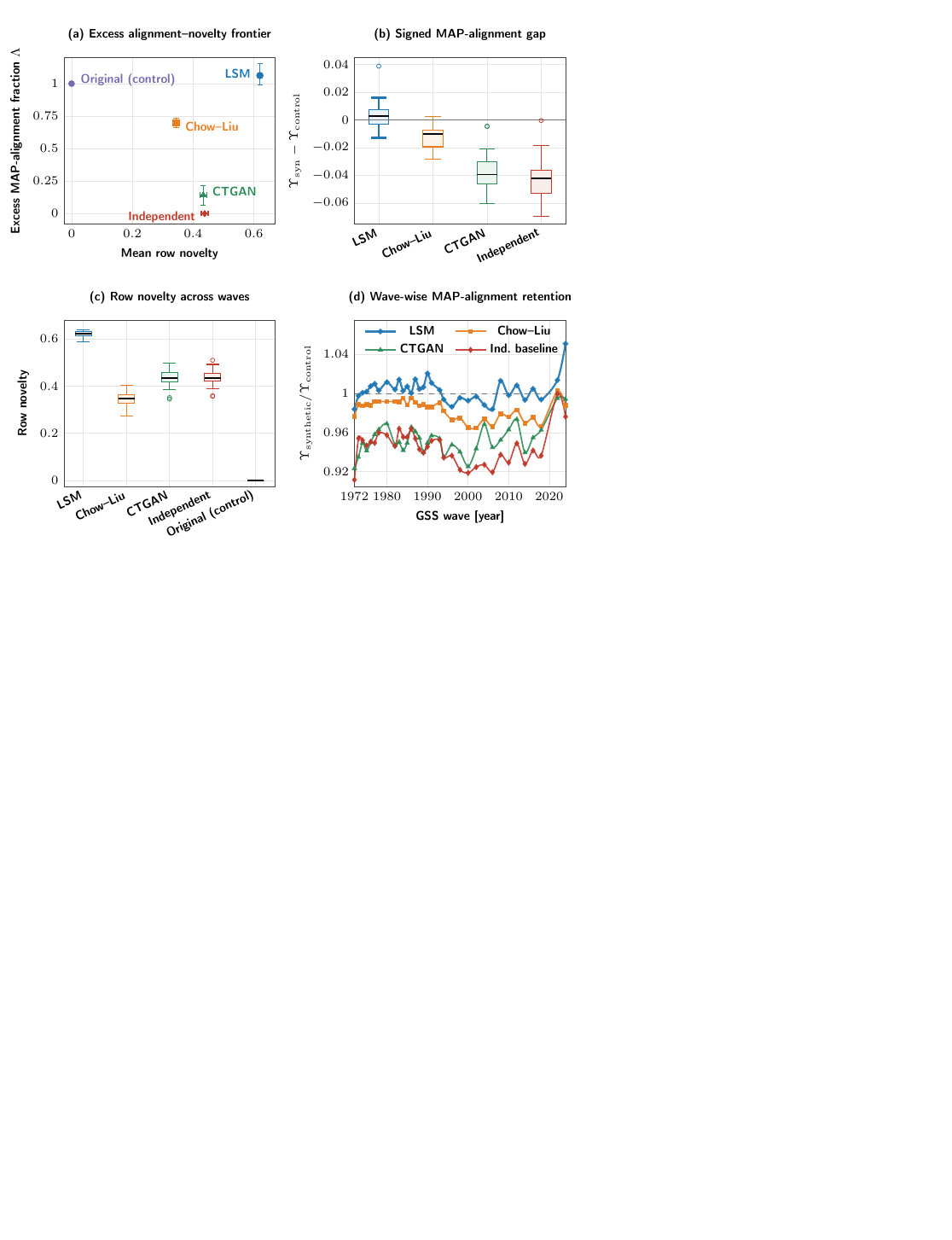}
     \fi

  \vspace{-540pt}
  
\caption{General Social Survey benchmark. Panel (a) shows excess MAP alignment relative to the independent-marginal baseline against mean row novelty. Panel (b) shows wave-level signed MAP-alignment gaps, panel (c) shows wave-level row novelty, and panel (d) shows wave-level retention. Intervals in panel (a) are nonparametric wave-bootstrap intervals. ``Chow--Liu'' denotes the restricted Chow--Liu hybrid.}
\label{fig:gss_benchmark_four_panel}
\end{figure}

\subsection{Influenza B genomic surveillance}
\label{sec:infb_results}

The influenza B analysis comprised seven overlapping sequence windows spanning 2016--2023. After alignment and preprocessing, individual sequences were represented by 220--374 categorical loci, depending on the window. This application tests a central ambiguity in surveillance synthesis: strong conditional agreement may reflect preservation of genomic dependence, but high similarity can also arise from repeating previously observed sequences. The windows overlap in time and are interpreted as temporally local surveillance analyses rather than independent population replicates. Results across the seven influenza B surveillance windows are summarized in Table~\ref{tab:infb_results}.

\begin{table}[!htp]
\centering
\caption{Influenza B results across seven overlapping genomic-surveillance windows. The output definitions are the same as in Table~\ref{tab:gss_benchmark}. Signed-gap intervals are descriptive run-resampling intervals.}
\label{tab:infb_results}
\small
\setlength{\tabcolsep}{3.2pt}
\resizebox{\textwidth}{!}{%
\begin{tabular}{lccccc}
\toprule
Generator & Mean MAP alignment & Signed gap (descriptive interval) & Retention & Nearest-record similarity & Row novelty \\
\midrule
LSM & 0.9804 & $-0.0070$ $[-0.0087,-0.0054]$ & 0.9929 & 0.6122 & 0.3878 \\
Chow--Liu hybrid & 0.9876 & 0.0002 $[-0.0003,0.0006]$ & 1.0002 & 0.9934 & 0.0066 \\
CTGAN & 0.9737 & $-0.0137$ $[-0.0165,-0.0105]$ & 0.9861 & 0.9185 & 0.0815 \\
Independent baseline & 0.9771 & $-0.0103$ $[-0.0154,-0.0049]$ & 0.9896 & 0.9487 & 0.0513 \\
Original-data control & 0.9874 & 0.0000 & 1.0000 & 1.0000 & 0.0000 \\
\bottomrule
\end{tabular}%
}
\end{table}

As shown in Table~\ref{tab:infb_results}, LSM retained 0.9929 of the original-data control MAP alignment and had row novelty 0.3878. The Chow--Liu hybrid had a mean score nearly identical to the control, but its mean nearest-record similarity was 0.9934, corresponding to novelty 0.0066. Thus, close agreement in the aggregate alignment score was accompanied by almost no separation from observed sequences. CTGAN and the independent baseline had both lower alignment and lower novelty than LSM. Panels (a) and (b) of Figure~\ref{fig:additional_domains} show the corresponding run-level distributions across the seven surveillance windows, highlighting the near-zero Chow--Liu alignment gap together with its consistently minimal novelty. The analysis did not evaluate the consequences of these differences for forecasting or emergence-risk estimation, which require separate prospective validation.

\subsection{National Social Life, Health, and Aging Project}
\label{sec:nshap_results}

NSHAP is a population-based longitudinal study of older adults in the United States, with measurements spanning physical and mental health, function, cognition, behavior, and social relationships \citep{jaszczak2014nshap,nshapround2}. Its heterogeneous schema provides a human-health setting in which synthetic records must preserve conditional relationships across clinical, functional, behavioral, and social domains rather than only variable-wise distributions. We analyzed eight archived evaluation instances from Round 2, each containing 861--862 modeled variables and 1,000 evaluation records per condition. Because the instances were derived from the same Round 2 source, they were used to assess stability of the archived analysis configuration rather than treated as independent cohorts. Results across the eight archived NSHAP analyses are summarized in Table~\ref{tab:nshap_results}.

As shown in Table~\ref{tab:nshap_results}, the mean LSM MAP-alignment score was 0.8294, compared with 0.7655 for the original-data control, and row novelty was 0.8474. The positive difference occurred in all eight archived analyses. Under the definition of MAP alignment, observed records contain both modal and nonmodal conditional realizations, whereas a generator may preferentially produce values near conditional modes. The NSHAP result is therefore consistent with conditional-mode concentration rather than superior recovery of the observed joint distribution. The Chow--Liu hybrid was closest to the control in mean alignment, while CTGAN and the independent baseline had negative signed gaps and intermediate novelty. Panels (c) and (d) of Figure~\ref{fig:additional_domains} show that the positive LSM alignment gap and high novelty persist across the archived NSHAP analyses, reinforcing the mode-concentration interpretation.

\begin{figure}[!ht] 
  \centering
  \tikzexternalenable
  \tikzsetnextfilename{figres2}
  \iftikzX
\else \includegraphics[width=1.6\textwidth]{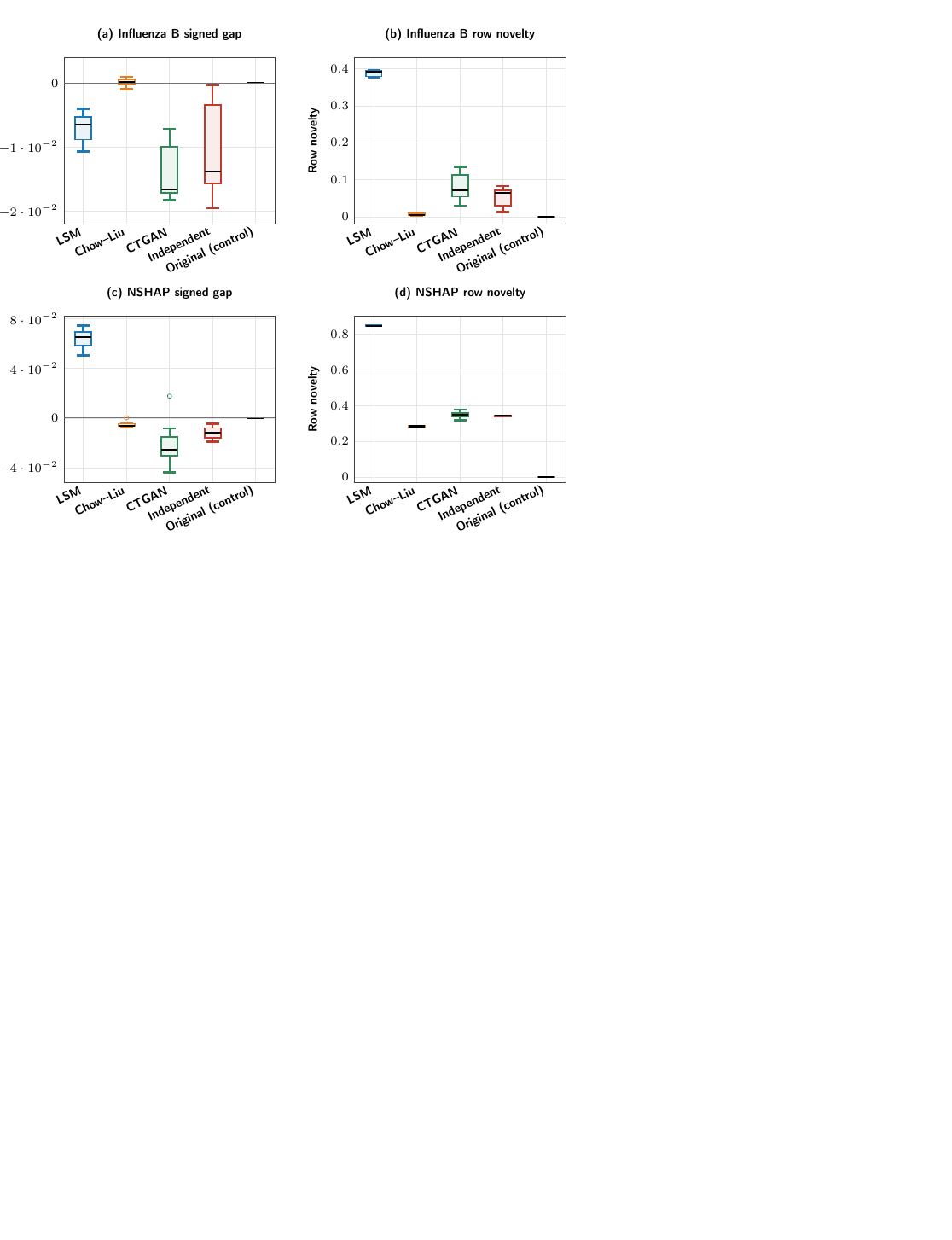}
    \fi

\vspace{-7in}
  
\caption{Run-level results for the additional applications. Panels (a) and (b) show signed MAP-alignment gaps and row novelty across seven overlapping influenza B windows. Panels (c) and (d) show the corresponding quantities across eight archived NSHAP Round 2 analyses. The original-data control is fixed at zero gap and zero novelty.}
\label{fig:additional_domains}
\end{figure}

\section{Discussion}

\begin{table}[!htp]
\centering
\caption{NSHAP Round 2 results across eight archived analysis instances. The output definitions are the same as in Table~\ref{tab:gss_benchmark}. Signed-gap intervals are descriptive run-resampling intervals.}
\label{tab:nshap_results}
\small
\setlength{\tabcolsep}{3.2pt}
\resizebox{\textwidth}{!}{%
\begin{tabular}{lccccc}
\toprule
Generator & Mean MAP alignment & Signed gap (descriptive interval) & Retention & Nearest-record similarity & Row novelty \\
\midrule
LSM & 0.8294 & 0.0639 $[0.0583,0.0692]$ & 1.0835 & 0.1526 & 0.8474 \\
Chow--Liu hybrid & 0.7602 & $-0.0053$ $[-0.0065,-0.0035]$ & 0.9931 & 0.7153 & 0.2847 \\
CTGAN & 0.7449 & $-0.0206$ $[-0.0313,-0.0079]$ & 0.9730 & 0.6521 & 0.3479 \\
Independent baseline & 0.7535 & $-0.0120$ $[-0.0155,-0.0085]$ & 0.9844 & 0.6575 & 0.3425 \\
Original-data control & 0.7655 & 0.0000 & 1.0000 & 1.0000 & 0.0000 \\
\bottomrule
\end{tabular}%
}
\end{table}

The methodological contribution of this work is a conditional-distribution approach to validating synthetic multivariate data. Full conditional distributions are directly tied to the joint law: under positivity and compatibility, the complete normalized conditional profile identifies the distribution, and its integrated $L^1$ difference defines a metric between finite-state generative processes. MAP alignment is the corresponding one-sided sample diagnostic. It asks whether records generated by an arbitrary synthesizer conform to conditional relationships estimated from held-out real data, without requiring access to the generator likelihood or selecting a downstream task in advance. Finite-sample concentration and consistency results make the statistic an estimable reference comparison rather than an informal plausibility score.

This distinction matters for biostatistics because the scientific utility of synthetic health and population data often depends on conditional rather than marginal validity. Regression coefficients, effect modification, risk stratification, subgroup comparisons, missing-data models, and transportability analyses can all be affected when cross-variable structure is distorted, even if univariate summaries are reproduced. A task-specific validation can establish fitness for the task that was chosen, but it cannot certify the data for analyses that were not anticipated. The proposed framework supplies a generator-agnostic structural diagnostic that can be applied before or alongside estimand-specific validation.

The empirical results illustrate why MAP alignment must be interpreted relative to an original-data control and jointly with record novelty. In GSS, LSM was compatible with the control in mean alignment while maintaining substantial separation from observed respondents, the desired pattern for preserving learned conditional structure without simple row reuse. In influenza B, Chow--Liu closely matched the control alignment but produced records with almost no novelty, showing that a favorable aggregate alignment score can coexist with near-reproduction of observed sequences. In NSHAP, LSM exceeded the control alignment mean. Because real records include legitimate nonmodal realizations, whereas a generator can preferentially produce conditional modes, an above-control score is not evidence of superior recovery of the joint distribution; together with high novelty, it indicates conditional-mode concentration. The three applications therefore expose different failure modes that would be difficult to distinguish using a single scalar utility measure.

MAP alignment is consequently a reference-based diagnostic, not an objective to maximize and not a proof that two joint distributions are equal. Nearest-record similarity is likewise a memorization diagnostic rather than a formal disclosure-risk guarantee. For a specified clinical, epidemiologic, or surveillance estimand, these quantities should supplement marginal checks, task-specific operating characteristics, calibration analyses, and appropriate privacy assessments. Their role is to determine whether a synthetic dataset has retained an identifying layer of multivariate structure and whether apparent agreement is attributable to reuse or mode concentration before the data are used for downstream inference.

The framework is especially relevant to settings in which likelihoods are unavailable, variables are numerous and heterogeneous, and downstream analyses cannot be exhaustively enumerated. A common conditional evaluator makes the statistical target explicit and permits the same structural questions to be asked across health surveys, real-world data, and genomic surveillance while leaving application-specific utility to be assessed separately.

\section{Limitations}

The identification results require strict positivity and compatibility of the full conditional system. Separately fitted conditional learners need not be exactly compatible with a single joint distribution. MAP alignment remains a well-defined evaluator statistic in that setting, but interpreting $d_\mu$ as a metric between underlying processes requires the stated compatibility assumption.

The formal development is for finite product spaces. The empirical analyses use categorical or discretely represented variables, while the Gaussian examples use a continuous-density analogue. A general treatment of continuous and mixed-type variables remains to be developed.

The analysis units do not represent exchangeable population replicates. GSS waves are temporally ordered and were analyzed without survey weights, strata, or primary sampling units. Influenza B windows overlap in calendar time. NSHAP analyses share one Round 2 source and use the archived full schema, including design and administrative variables. The reported intervals summarize variation within the observed analyses rather than population-level sampling uncertainty. Results may also depend on missingness patterns and discretization.

Computational cost increases with the number of variables and evaluation records because the evaluator fits one conditional model per coordinate and scores each record-coordinate pair. Parallel computation made the present analyses tractable, but repeated fitting and substantially larger schemas may require additional approximations.

\section{Future work}

Further work should formalize the continuous and mixed-type extension, incorporate complex survey design into the reference measure and aggregation, and repeat the NSHAP analysis using a prespecified set of health and social variables. Evaluator sensitivity should be studied through repeated data partitions and conditional-model fits. Empirical reports should include duplicate rates and upper-tail nearest-record similarities, and genomic-surveillance studies should assess whether the proposed diagnostics predict prospective forecasting or emergence-risk performance.

\section{Conclusion}

Synthetic-data validation should determine whether the joint dependence structure of the target population is retained and whether apparent fidelity is obtained by reproducing source records. The normalized full-conditional profile provides an identifying representation and a metric on finite-state generative processes, while MAP alignment supplies an estimable one-sided comparison against held-out real data. Paired with nearest-record similarity, the framework distinguished loss of dependence, near-reuse, and conditional-mode concentration across human-health, genomic-surveillance, and repeated-survey applications. These diagnostics provide a theory-grounded structural validation layer for biostatistical uses of synthetic data and are intended to complement estimand-specific utility and disclosure-risk analyses.

\section{Data and code availability}

The GSS data are publicly available from NORC \citep{gss2024}. NSHAP Round 2 data and documentation are distributed through ICPSR \citep{nshapround2}. Processed influenza B evaluation outputs and run metadata are included in the reproducibility materials. The LSYNTH implementation and examples are available at \url{https://github.com/zeroknowledgediscovery/lsynth}; the Python package can be installed with \texttt{pip install lsynth}.

\bibliographystyle{plainnat}
\bibliography{lsyn}

@article{brook1964,
  author    = {Brook, D.},
  title     = {On the distinction between conditional and unconditional distributions},
  journal   = {Biometrika},
  volume    = {51},
  number    = {3--4},
  pages     = {481--483},
  year      = {1964}
}

@article{dobrushin1968,
  author    = {Dobrushin, R. L.},
  title     = {The description of a random field by means of conditional probabilities and conditions of its regularity},
  journal   = {Theory of Probability and its Applications},
  volume    = {13},
  number    = {2},
  pages     = {197--224},
  year      = {1968}
}

@incollection{hammersley1971,
  author    = {Hammersley, J. M. and Clifford, P.},
  title     = {Markov fields on finite graphs and lattices},
  booktitle = {Markov Random Fields},
  editor    = {Grimmett, P.},
  publisher = {Springer},
  year      = {2017},
  note      = {Originally written in 1971 as an unpublished manuscript}
}

@article{besag1974,
  author    = {Besag, J.},
  title     = {Spatial interaction and the statistical analysis of lattice systems},
  journal   = {Journal of the Royal Statistical Society, Series B},
  volume    = {36},
  number    = {2},
  pages     = {192--236},
  year      = {1974}
}

@article{hothorn2006,
  author    = {Hothorn, T. and Hornik, K. and Zeileis, A.},
  title     = {Unbiased recursive partitioning: A conditional inference framework},
  journal   = {Journal of Computational and Graphical Statistics},
  volume    = {15},
  number    = {3},
  pages     = {651--674},
  year      = {2006}
}

@article{hothorn2006party,
  author    = {Hothorn, T. and Hornik, K. and Zeileis, A.},
  title     = {{party}: A laboratory for recursive partitioning},
  journal   = {R News},
  volume    = {6},
  number    = {2},
  pages     = {17--23},
  year      = {2006}
}

@article{strobl2007,
  author    = {Strobl, C. and Boulesteix, A. and Hothorn, T. and Zeileis, A.},
  title     = {Bias in random forest variable importance measures: Illustrations, sources and a solution},
  journal   = {BMC Bioinformatics},
  volume    = {8},
  number    = {25},
  pages     = {1--21},
  year      = {2007}
}

@article{hoeffding1963,
  author    = {Hoeffding, W.},
  title     = {Probability inequalities for sums of bounded random variables},
  journal   = {Journal of the American Statistical Association},
  volume    = {58},
  number    = {301},
  pages     = {13--30},
  year      = {1963}
}

@inproceedings{kingma2014autoencoding,
  author    = {Diederik P. Kingma and Max Welling},
  title     = {Auto-Encoding Variational Bayes},
  booktitle = {Proc. International Conference on Learning Representations (ICLR)},
  year      = {2014},
  note      = {arXiv:1312.6114}
}

@article{jelinek1977perplexity,
  author  = {Frederick Jelinek and Robert L. Mercer and Lalit R. Bahl and James K. Baker},
  title   = {Perplexity---a Measure of the Difficulty of Speech Recognition Tasks},
  journal = {Journal of the Acoustical Society of America},
  volume  = {62},
  number  = {S1},
  pages   = {S63},
  year    = {1977},
  doi     = {10.1121/1.2016299}
}

@article{esteban2017realvalued,
  author  = {Crist{\'o}bal Esteban and Stephanie L. Hyland and Gunnar R{\"a}tsch},
  title   = {Real-valued (Medical) Time Series Generation with Recurrent Conditional {GANs}},
  journal = {arXiv preprint arXiv:1706.02633},
  year    = {2017}
}

@inproceedings{heusel2017gans,
  author    = {Martin Heusel and Hubert Ramsauer and Thomas Unterthiner and Bernhard Nessler and Sepp Hochreiter},
  title     = {{GAN}s Trained by a Two Time-Scale Update Rule Converge to a Local {N}ash Equilibrium},
  booktitle = {Advances in Neural Information Processing Systems (NeurIPS)},
  year      = {2017},
  note      = {Introduces the Fr\'echet Inception Distance (FID)}
}

@inproceedings{lopezpaz2017revisiting,
  author    = {David Lopez-Paz and Maxime Oquab},
  title     = {Revisiting Classifier Two-Sample Tests},
  booktitle = {Proc. International Conference on Learning Representations (ICLR)},
  year      = {2017}
}

@article{snoke2018general,
  author  = {Joshua Snoke and Gillian M. Raab and Beata Nowok and Chris Dibben and Aleksandra Slavkovic},
  title   = {General and Specific Utility Measures for Synthetic Data},
  journal = {Journal of Privacy and Confidentiality},
  volume  = {8},
  number  = {1},
  year    = {2018}
}

@article{nowok2016synthpop,
  author  = {Beata Nowok and Gillian M. Raab and Chris Dibben},
  title   = {synthpop: Bespoke Creation of Synthetic Data in {R}},
  journal = {Journal of Statistical Software},
  volume  = {74},
  number  = {11},
  year    = {2016}
}

@misc{gss2024,
  author       = {{NORC at the University of Chicago}},
  title        = {General Social Survey: 1972--2024 Cumulative Data File},
  year         = {2024},
  howpublished = {General Social Survey},
  url          = {https://gss.norc.org/}
}

@article{sizemore2024digital,
  author  = {Sizemore, Nicholas and Oliphant, Kaitlyn and Zheng, Ruolin and Martin, Camilia R. and Claud, Erika C. and Chattopadhyay, Ishanu},
  title   = {A Digital Twin of the Infant Microbiome to Predict Neurodevelopmental Deficits},
  journal = {Science Advances},
  volume  = {10},
  number  = {15},
  pages   = {eadj0400},
  year    = {2024},
  doi     = {10.1126/sciadv.adj0400}
}

@article{wu2026emergenet,
  author  = {Wu, Kevin and Li, Feng and Chattopadhyay, Ishanu},
  title   = {Emergenet: A Digital Twin of Influenza A Evolution for Vaccine Strain Forecasting and Emergence Risk Assessment},
  journal = {Military Medicine},
  year    = {2026},
  doi     = {10.1093/milmed/usag255},
  note    = {In print}
}

@article{chow1968,
  author  = {Chow, C. K. and Liu, C. N.},
  title   = {Approximating Discrete Probability Distributions with Dependence Trees},
  journal = {IEEE Transactions on Information Theory},
  volume  = {14},
  number  = {3},
  pages   = {462--467},
  year    = {1968},
  doi     = {10.1109/TIT.1968.1054142}
}

@inproceedings{xu2019ctgan,
  author    = {Xu, Lei and Skoularidou, Maria and Cuesta-Infante, Alfredo and Veeramachaneni, Kalyan},
  title     = {Modeling Tabular Data Using Conditional {GAN}},
  booktitle = {Advances in Neural Information Processing Systems},
  volume    = {32},
  year      = {2019}
}

@article{jaszczak2014nshap,
  author  = {Jaszczak, Angela and O'Doherty, Katie and Colicchia, Michael and Satorius, Jennifer and McPhillips, Jane and Czaplewski, Michael and Smith, Stephen},
  title   = {Continuity and Innovation in the Data Collection Protocols of the Second Wave of the National Social Life, Health, and Aging Project},
  journal = {The Journals of Gerontology: Series B},
  volume  = {69},
  number  = {Suppl 2},
  pages   = {S4--S14},
  year    = {2014},
  doi     = {10.1093/geronb/gbu031}
}

@misc{nshapround2,
  author       = {Waite, Linda J. and Cagney, Kathleen A. and Dale, William and Huang, Elbert S. and Laumann, Edward O. and McClintock, Martha K. and Cornwell, Benjamin},
  title        = {National Social Life, Health, and Aging Project ({NSHAP}): Round 2 and Partner Data Collection, United States, 2010--2011},
  year         = {2023},
  howpublished = {Inter-university Consortium for Political and Social Research},
  doi          = {10.3886/ICPSR34921.v5}
}

\end{document}